\documentclass{article}
\usepackage[utf8]{inputenc}
\usepackage[USenglish]{babel} 

\usepackage{authblk} 
\usepackage{fullpage}
\usepackage{csquotes}    

\usepackage{setspace}
\usepackage{geometry}
\usepackage{nccmath}

\usepackage{graphicx} 

\usepackage{amsthm, amsmath, amssymb}

\usepackage{algorithm}
\usepackage[noend]{algpseudocode}

\usepackage{hyperref}
\hypersetup{
    colorlinks=true,
    linkcolor=[rgb]{0.5,0,0},
    filecolor=magenta,
    citecolor=[rgb]{0.7,0.3,0},
    urlcolor=[rgb]{0,0,0.5},
    linktocpage=true, 
    hypertexnames=false, 
}

\usepackage{todonotes}

\usepackage[nameinlink, capitalise]{cleveref}

\usepackage{thmtools}
\usepackage{thm-restate}

\usepackage{tcolorbox}

\theoremstyle{plain}
\newtheorem{theorem}{Theorem}
\newtheorem{lemma}[theorem]{Lemma}

\newtheorem{proposition}[lemma]{Proposition}

\theoremstyle{definition}
\newtheorem{definition}[theorem]{Definition} 
\newtheorem{observation}[theorem]{Observation}

\newtheorem*{remark*}{Remark}

\crefname{claim}{claim}{claims}
\crefname{assumption}{assumption}{assumptions}

\title{Improved Weighted Matching in the Sliding Window Model\footnote{C.A. and K.K.N. are supported by EPSRC DTP studentship EP/T517872/1. P.D. and C.K. are supported by EPSRC New Investigator Award EP/V010611/1. P.D. is also supported by Czech Science Foundation GAČR grant \#22-14872O. }} 

\author[1]{Cezar-Mihail Alexandru}

\author[2]{Pavel Dvořák}

\author[3]{Christian Konrad}

\author[4]{Kheeran K. Naidu}

\affil[1]{ca17021@bristol.ac.uk}
\affil[2]{koblich@iuuk.mff.cuni.cz}
\affil[3]{christian.konrad@bristol.ac.uk}
\affil[4]{kheeran.naidu@bristol.ac.uk}

\affil[{ }]{}

\affil[1,3,4]{Department of Computer Science, University of Bristol, UK}
\affil[2]{Tata Institute of Fundamental Research, Mumbai, India} 
\affil[2]{Faculty of Mathematics and Physics, Charles University,
Prague, Czech Republic}

\date{}

\newcommand{\Ot}{\ensuremath{\tilde{O}}}

\DeclareMathOperator{\polylog}{\text{polylog}}

\newcommand{\Stack}{{\textsf{Stack}}}

\newcommand{\R}{\mathbb{R}}
\newcommand{\alg}{{\cal ALG}}

\newcommand{\cA}{{\cal A}}
\newcommand{\cI}{{\cal I}}
\newcommand{\mwm}{$\mathcal{ALG}^\varepsilon_\textit{PS}$}
\newcommand{\w}[1]{w\bigl(#1\bigr)}

\makeatletter
\DeclareFontFamily{U}{MnSymbolA}{}
\DeclareFontShape{U}{MnSymbolA}{m}{n}{
    <-6>  MnSymbolA5
   <6-7>  MnSymbolA6
   <7-8>  MnSymbolA7
   <8-9>  MnSymbolA8
   <9-10> MnSymbolA9
  <10-12> MnSymbolA10
  <12->   MnSymbolA12}{}
\DeclareFontShape{U}{MnSymbolA}{b}{n}{
    <-6>  MnSymbolA-Bold5
   <6-7>  MnSymbolA-Bold6
   <7-8>  MnSymbolA-Bold7
   <8-9>  MnSymbolA-Bold8
   <9-10> MnSymbolA-Bold9
  <10-12> MnSymbolA-Bold10
  <12->   MnSymbolA-Bold12}{}
\DeclareSymbolFont{MnSyA}{U}{MnSymbolA}{m}{n}
\SetSymbolFont{MnSyA}{bold}{U}{MnSymbolA}{b}{n}

\DeclareRobustCommand{\overleftharpoon}{\mathpalette{\overarrow@\leftharpoonfill@}}
\DeclareRobustCommand{\overrightharpoon}{\mathpalette{\overarrow@\rightharpoonfill@}}
\def\leftharpoonfill@{\arrowfill@\leftharpoondown\mn@relbar\mn@relbar}
\def\rightharpoonfill@{\arrowfill@\mn@relbar\mn@relbar\rightharpoonup}

\DeclareMathSymbol{\leftharpoondown}{\mathrel}{MnSyA}{'112}
\DeclareMathSymbol{\rightharpoonup}{\mathrel}{MnSyA}{'100}
\DeclareMathSymbol{\mn@relbar}{\mathrel}{MnSyA}{'320}
\makeatother

\begin{document}

\maketitle

\begin{abstract}
We consider the \textsf{Maximum-weight Matching} (\textsf{MWM}) problem in the streaming sliding window model of computation. 
In this model, the input consists of a sequence of weighted edges on a given vertex set $V$ of size $n$.  
The objective is to maintain an approximation of a maximum-weight matching in the graph spanned by the $L$ most recent edges, for some integer $L$, using as little space as possible. 
Prior to our work, the state-of-the-art results were a $(3.5+\varepsilon)$-approximation algorithm for \textsf{MWM} by Biabani et al. [ISAAC'21] and a $(3+\varepsilon)$-approximation for (unweighted) \textsf{Maximum Matching} (\textsf{MM}) by Crouch et al.\ [ESA'13]. Both algorithms use space $\Ot(n)$. 

We give the following results:

\begin{enumerate}
 \item We give a $(2+\varepsilon)$-approximation algorithm for \textsf{MWM} with space $\Ot(\sqrt{nL})$. Under the reasonable assumption that the graphs spanned by the edges in each sliding window are simple, our algorithm uses space $\Ot(n \sqrt{n})$. 
 
 \item In the $\Ot(n)$ space regime, we give a $(3+\varepsilon)$-approximation algorithm for \textsf{MWM}, thereby closing the gap between the best-known approximation ratio for \textsf{MWM} and \textsf{MM}. 
\end{enumerate}

Similar to Biabani et al.'s \textsf{MWM} algorithm, both our algorithms execute multiple instances of the $(2+\varepsilon)$-approximation $\Ot(n)$-space streaming algorithm for \textsf{MWM} by Paz and Schwartzman [SODA'17] on different portions of the stream. Our improvements are obtained by selecting these substreams differently. Furthermore, our $(2+\varepsilon)$-approximation algorithm runs the Paz-Schwartzman algorithm in {\em reverse direction} over some parts of the stream, and in {\em forward direction} over other parts, which allows for an improved approximation guarantee at the cost of increased space requirements.

\end{abstract}

\newpage 
\section{Introduction}
The {\em data streaming model} is a well-established computational model that provides a framework for studying massive data set algorithms. 
The defining features of the model are restricted access to the input data and sublinear space. 
A data streaming algorithm processes its input sequentially in a single pass while maintaining only a small summary of the input in memory.

In this paper, we study the \textsf{Maximum-weight Matching} (\textsf{MWM}) problem in the {\em (streaming)  sliding window} model. 
In this variant of the streaming model, the input consists of a potentially infinite sequence $e_1, e_2, \dots$ of weighted edges on an underlying vertex set $V$ of size $n$. 
The objective is to maintain a matching of large weight in the graph spanned by the $L$ most recent edges, for some integer $L$, using as little space as possible. 
In more detail, after having processed the current edge $e_i$, for every $i$, the objective is to report an approximation of a maximum-weight matching in the graph spanned by the current sliding window $E_i := \{ e_j \ : \ \max \{i - L+1, 1 \} \le j \le i \}$. Many of the known sliding window algorithms for graph problems operate within {\em semi-streaming space} \cite{Feigenbaum2005OnGP}, i.e., within space $O(n \polylog n) = \Ot(n)$. In this paper, we will work both with the semi-streaming space regime and also consider algorithms that use more space. 

While sliding window algorithms have been studied for two decades~\cite{Datar2002MaintainingSS}, sliding window algorithms for graph problems were first
 considered by Crouch et al.~\cite{cms13} in 2013. 
Amongst other results, they showed that there is a $(3+\varepsilon)$-approximation semi-streaming sliding window algorithm for unweighted \textsf{Maximum Matching} (\textsf{MM}) and a $9.027$-approximation semi-streaming sliding window algorithm for $\textsf{MWM}$. 
While no improved results are known for \textsf{MM}, Crouch and Stubbs~\cite{Crouch2014ImprovedSA} subsequently improved upon the result for \textsf{MWM} and gave a $(6+\varepsilon)$-approximation semi-streaming algorithm, and, very recently, Biabani et al.~\cite{Biabani21_WeightedMatching} gave a $(3.5+\varepsilon)$-approximation in the semi-streaming space regime. 
The state-of-the-art results for \textsf{MM} and \textsf{MWM} in the semi-streaming sliding window model therefore do not yet line up.

\paragraph*{Our Results}
In this paper, we give two sliding window algorithms for \textsf{MWM} that both improve upon the state-of-the-art approximation guarantee of $3.5+\varepsilon$. 

As our first result, we give a substantial improvement and obtain an approximation factor of $2+\varepsilon$ at the expense of increased space requirements:

\begin{theorem}[simplified version]
\label{thm:MainResult2}
 There is a deterministic $(2+\varepsilon)$-approximation sliding window algorithm for \textsf{Maximum-weight Matching} that uses space $\Ot(\sqrt{n L})$ (with dependency on $\varepsilon$ and logarithms suppressed), for any $\varepsilon > 0$.
\end{theorem}
\newcounter{counterMainThm2} 
\setcounter{counterMainThm2}{\value{theorem}}

Some remarks are in order. First, we observe that going beyond the approximation factor of $2$, even using space $O(n^{1.999})$, would answer a long-standing open problem in graph streaming research, namely, whether there is a one-pass $(2-\Omega(1))$-approximation streaming algorithm for \textsf{MM} with space $O(n^{1.999})$. We thus cannot expect to obtain further improvements in the approximation guarantee with current techniques. 
Second, the space requirements of our algorithm depend on the sliding window length $L$. This is in contrast to the $(3.5+\varepsilon)$-approximation algorithm by Biabani et al.~\cite{Biabani21_WeightedMatching} and our second algorithm described below. Under the natural assumption that the graphs described by all sliding windows are simple, we have $L = O(n^2)$, which yields a space bound of $\Ot(n \sqrt{n})$.

As our second result, we close the gap between \textsf{MM} and \textsf{MWM} in the semi-streaming space regime. 
To this end, we give a semi-streaming sliding window algorithm for \textsf{MWM} that matches the approximation guarantee of the best-known sliding window algorithm for \textsf{MM}.

\begin{theorem}[simplified version]
\label{thm:MainResult}
 There is a deterministic semi-streaming sliding window algorithm for \textsf{Maximum-weight Matching} with approximation factor $3+\varepsilon$, for any $\varepsilon > 0$.
\end{theorem}
\newcounter{counterMainThm} 
\setcounter{counterMainThm}{\value{theorem}}

Table~\ref{tab:sliding-results} summarizes all results known for \textsf{MM} and \textsf{MWM} in the sliding window model.

\begin{table}
    \centering
    \begin{tabular}[t]{l|c|c|l}
& Approximation Factor & Space & Reference \\
\hline
\textsf{MM} &  $3 + \varepsilon$ & $\Ot(n)$ & Crouch et al.\ \cite{cms13} \\
\hline
\textsf{MWM} & $9.027$ & $\Ot(n)$ & Crouch et al.\ \cite{cms13} \\
 & $6 + \varepsilon$ & $\Ot(n)$ & Crouch and Stubbs \cite{Crouch2014ImprovedSA} \\
 & $3.5 + \varepsilon$ & $\Ot(n)$ & Biabani et al.\ \cite{Biabani21_WeightedMatching} \\
 & $3 + \varepsilon$ & $\Ot(n)$ &  This paper (Theorem~\ref{thm:MainResult}) \\
 & $2 + \varepsilon$ & $\Ot(\sqrt{n L})$ &  This paper (Theorem~\ref{thm:MainResult2})
\end{tabular} 
   
\smallskip
\caption{Known sliding window algorithms for \textsf{MM} and \textsf{MWM}.\label{tab:sliding-results}} 
\end{table}

\paragraph*{Techniques}
\label{ssec:Techniques}
Both our algorithms make use of the one-pass $(2+\varepsilon)$-approximation streaming algorithm for \textsf{MWM} by Paz and Schwartzman \cite{ps17}. Since we make use of the inner workings of the algorithm, we will discuss this algorithm first. 

\vspace{0.1cm}
\textit{Paz and Schwartzman's \textsf{MWM} Algorithm.} 
Paz and Schwartzman's original algorithm \cite{ps17} uses space $O\bigl(\frac{1}{\varepsilon} \cdot n \log^2 n\bigr)$ and is based on the \textit{local ratio technique} (see \cite{bbfr04} for further details on this technique). 
Ghaffari and Wajc~\cite{gw19} gave a simplified version and improved the space complexity to the (optimal in $n$) bound $O\bigl(\frac{\log(1/\varepsilon)}{\varepsilon} \cdot n \log n\bigr)$.

The Paz and Schwartzman algorithm with Ghaffari and Wajc's improvement works as follows. 
For every vertex $v \in V$, it maintains a potential $\varphi(v)$ that is initialized with $0$, and uses a stack data structure \textsf{Stack}. 
When an edge $e=\{u, v\}$ arrives in the stream, $e$ is pushed onto \textsf{Stack} if its weight $w(e)$ exceeds the sum of the potentials of its incident vertices by a factor of at least $(1+\varepsilon)$, i.e., $w(e) \ge (1+\varepsilon)(\varphi(u) + \varphi(v))$. 
The discrepancy between $w(e)$ and $\varphi(u) + \varphi(v)$ is denoted the {\em reduced weight} of $e$ and is abbreviated by $w'(e) := w(e) - (\varphi(u) + \varphi(v))$. 
Then, the potentials $\varphi(u)$ and $\varphi(v)$ are updated as $\varphi(u) = \varphi(u) + w'(e)$ and $\varphi(v) = \varphi(v) + w'(e)$. Last, if either $u$ or $v$ is adjacent to at least $\frac{3 \log(1/\varepsilon)}{\varepsilon} + 1$ edges in \textsf{Stack} then the oldest (and thus lightest) edge incident to the vertex is removed from \textsf{Stack}, thereby limiting the number of edges on \textsf{Stack}.
After having processed all the edges in the stream, the output matching $\hat{M}$ is computed in a post-processing step. 
The edges in \textsf{Stack} are popped one by one and greedily inserted into $\hat{M}$ if possible, i.e., as long as $\hat{M}$ remains a matching. We denote the Paz and Schwartzman algorithm by {\mwm}.
See \cref{sec:Preliminaries} for a formal description.

\vspace{0.1cm}
\textit{$(2+\varepsilon)$-approximation Algorithm with Space $\Ot(\sqrt{nL})$.} Our $(2+\varepsilon)$-approximation algorithm processes the input in blocks of size $s = \tilde{\Theta}(\sqrt{nL})$. Consider one such block $B_j$, i.e., a substream of $s$ consecutive edges. The key idea of our algorithm is to run multiple instances of the Paz-Schwartzman algorithm {\mwm} on $B_j$, however, in {\em reverse direction}. We start with a single instance $\cI_1$. At various moments during the processing of $B_j$, we fork the current instance $\cI_i$ to obtain an additional instance $\cI_{i+1}$, and then only continue to feed further edges into $\cI_{i+1}$; thus, in any moment of processing the block $B_j$, we feed the edge to only one instance of the Paz-Schwartzman algorithm. The fork happens when the sum of reduced weights $W'(\cI_{i})$ of the edges on \textsf{Stack} in $\cI_{i}$ exceeds the sum of reduced weights of the previous instance by a $1+\varepsilon$ factor, i.e., $W'(\cI_{i}) > (1+\varepsilon) \cdot W'(\cI_{i-1})$. 
As a result, we obtain instances of Paz-Schwartzman that processed suffixes of different lengths of block $B_j$ (remember that we process $B_j$ in the reverse direction), and adjacent instances have a similar sum of reduced weights (up to a $1+\varepsilon$ factor). As we will point out in Section~\ref{sec:Preliminaries}, the sum of reduced weights in an instance of Paz-Schwartzman is strongly related to the weight of a maximum-weight matching among the edges observed thus far, and we heavily exploit this property in our proofs.

In each block $B_j$, besides preparing the instances of Paz-Schwartzman as described above, we also feed the edges of $B_j$ (in the forward direction) into those instances of Paz-Schwartzman that were prepared during previous blocks $B_{j'}$, with $j' < j$, and that are still {\em alive}, i.e., have only been fed edges from the current sliding window. As such, each instance of Paz-Schwartzman is executed on a portion of the stream in the reverse direction, followed by all the subsequent edges from more recent blocks in the forward direction until the current edge. The output produced when processing the current edge is the output of the oldest alive instance of Paz-Schwartzman.

Consider two adjacent instances $\cI_{i}$ and $\cI_{i+1}$ of Paz-Schwartzman prepared in the same block, where $\cI_i$ has processed only a subset of the edges of $\cI_{i+1}$ and their sums of reduced weights $W'$ are such that $W'(\cI_{i+1}) \approx (1+\varepsilon) W'(\cI_{i})$. The key benefit of executing Paz-Schwartzman in the reverse direction as opposed to forward is that the edges processed by $\cI_{i+1}$ but not by $\cI_{i}$ contribute to the sum of reduced weights only with an $\varepsilon$-fraction of $W'(\cI_{i})$ (since  $W'(\cI_{i+1}) - W'(\cI_{i}) \approx \varepsilon W'(\cI_i)$). When $\cI_i$ is the oldest alive instance and thus constitutes the output of our algorithm, we only {\em miss} an $\varepsilon$-fraction in terms of reduced weights of the edges in the sliding window that $\cI_{i}$ has not considered. We remark that this property could not be established if we run Paz-Schwartzman in the forward direction. 
This property together with the fact that the sum of reduced weights is related to the weight of a maximum-weight matching allows us to establish the approximation factor of our algorithm. 

Since only the $L$ most recent edges are relevant, our algorithm considers at most $\frac{L}{s} = \tilde{\Theta}(\sqrt{L/n})$ blocks simultaneously. Each block consists of $\Ot(1)$ instances of Paz-Schwartzman. Since each of these instances requires space $\Ot(n)$, we obtain the final space bound of $\Ot(n) \cdot \frac{L}{s} = \Ot(\sqrt{n L})$.

\vspace{0.1cm}
\textit{$(3+\varepsilon)$-approximation Semi-streaming Algorithm.} 
Our $(3+\varepsilon)$-approximation algorithm follows similar arguments as the $(3.5+\varepsilon)$-approximation algorithm by Biabani et al.~\cite{Biabani21_WeightedMatching}. We will therefore first explain the techniques behind Biabani et al.'s algorithm and then discuss our new ideas which yield the improved approximation guarantee.  

Biabani et al.'s algorithm combines the {\em smooth histogram technique} for sliding window algorithms by Braverman and Ostrovsky~\cite{Braverman2007SmoothHF} with the Paz and Schwartzman algorithm.
Braverman and Ostrovsky showed that if a function $f$ fulfills certain smoothness criteria\footnote{Informally, a function $f: 2^X \to \R$ is considered to be {\em smooth} if it satisfies the following: If $f(A)$ is close to $f(B)$ for $A,B \subseteq X$, for a suitable notion of closeness, then the values $f(A \cup C)$ and $f(B \cup C)$ are close for all $C \subseteq X$.} then a sliding window algorithm for approximating $f$ can be obtained from a traditional (non-sliding window) streaming algorithm for $f$ at the expense of only a logarithmic increase in the space requirements (as long as the approximation factor of the streaming algorithm is constant), and a slight increase in the approximation factor. 
In the context of \textsf{MWM}, the smoothness criteria are captured by Biabani et al.~\cite{Biabani21_WeightedMatching} via the notion of {\em lookahead} algorithm.
\begin{definition}[$(f, \alpha, \beta)$-lookahead algorithm~\cite{Biabani21_WeightedMatching}] \label{def:lookahead}
Let $\beta \in (0,1)$ and $\alpha > 0$ be real numbers. 
Let $X$ be a ground set, $S$ a stream of items of $X$, and let $f: 2^X \rightarrow \mathbb{R}^+$ be a non-decreasing function.
We say that a streaming algorithm $\mathcal{ALG}$ is a {\em $(f,\alpha, \beta)$-lookahead algorithm} if, for any partitioning of $S$ into three substreams $A,B,C$ with $\mathcal{ALG}(B) \ge (1-\beta) \cdot \mathcal{ALG}(AB)$, the following holds: $f(ABC) \le \alpha \cdot \mathcal{ALG}(BC)$.
\end{definition}
In this paper, the stream $AB$ denotes the concatenation of streams $A$ and $B$ (as it is used in the previous definition).
We observe that the previous definition holds for real-valued non-decreasing functions. 
In the context of \textsf{MWM}, the weight of a maximum-weight matching rather than the matching itself fulfills these conditions. 
We will therefore consider the problem of determining the weight of a maximum-weight matching instead, and, in order to be able to output an actual matching as required in \textsf{MWM}, we will rely on the fact that the underlying algorithm which we will consider also maintains the actual matching itself. Furthermore, we will write $\textsf{MWM}(S)$ to denote the weight of a maximum-weight matching in stream $S$.

Biabani et al.~\cite{Biabani21_WeightedMatching} showed that if there is a $(\textsf{MWM}, \alpha, \beta)$-lookahead algorithm that uses space $s$ then there exists a sliding-window algorithm with approximation ratio $\alpha$ and space $O\bigl(\frac{1}{\beta} \cdot s \log \sigma\bigr)$, where $\sigma = \frac{n}{2} \cdot w_{\text{max}} / w_{\text{min}}$ and $w_{\text{max}}$ and $w_{\text{min}}$ are the maximum and minimum weights of an edge in the input stream, respectively. 
Observe that, under the usual assumption that $w_{\text{max}} / w_{\text{min}}$ is polynomial in $n$, we have $\log \sigma = O(\log n)$.

The main part of their analysis is to show that a monotonic version of the Paz and Schwartzman algorithm, denoted $\alg_{\textit{mon}}$, constitutes a $(\textsf{MWM}, (3.5 +  \varepsilon), \beta)$-lookahead algorithm, for small values of $\varepsilon$ and $\beta \le \varepsilon / 9$. 
Combined, this yields a $(3.5+\varepsilon)$-approximation semi-streaming sliding window algorithm for \textsf{MWM}.

We first note (see \Cref{sec:lb} for details) that the analysis of Biabani et al.\ is best possible in that the Paz and Schwartzman algorithm and its monotonic version are no better than $(\textsf{MWM}, 3.5, \beta)$-lookahead algorithms. 
The smooth histogram technique applied to lookahead algorithms as defined in \Cref{def:lookahead} thus cannot give an improved approximation guarantee when Paz and Schwartzman's algorithm is used as the underlying algorithm.  
 
To illustrate our improvement, we first provide insight into the structure of Biabani et al.'s analysis. 
In order to prove that $\alg_{\textit{mon}}$ is a $(\textsf{MWM}, 3.5 + \varepsilon, \beta)$-lookahead algorithm, Biabani et al. relate $\textsf{MWM}(ABC)$ to the output of $\alg_{\textit{mon}}$ on various substreams of $ABC$:
\begin{eqnarray}
\textsf{MWM}(ABC) & \le & 2(1+\varepsilon)\cdot \bigl( \alg_{\textit{mon}}(AB) + \alg_{\textit{mon}}(BC) \bigr)  \nonumber \\
& & - \  \frac{1}{2(1+\varepsilon)}\cdot \alg_{\textit{mon}}(B) \ .   \label{eqn:392}
\end{eqnarray}
They subsequently use the smoothness assumption from Definition~\ref{def:lookahead} and a monotonicity property of $\alg_{\textit{mon}}$ to relate $\alg_{\textit{mon}}(AB)$ and $\alg_{\textit{mon}}(B)$ to $\alg_{\textit{mon}}(BC)$. 
This ultimately yields the desired bound $\textsf{MWM}(ABC) \le  (3.5 + \varepsilon) \cdot \alg_{\textit{mon}}(BC)$.

To obtain our improvement, we observe that a similar inequality to Inequality~\ref{eqn:392} can be obtained by considering {\em sums of reduced weights} of the respective runs of $\alg_{\textit{mon}}$ instead of the weights of the output matchings of $\alg_{\textit{mon}}$ on the different substreams. 
This idea is motivated by the fact that the sum of reduced weights is a lower bound on the weight of the matching produced by the algorithm, which can therefore give a more precise analysis. 
However, when departing from such an inequality involving sums of reduced weights, we unfortunately cannot immediately complete our analysis since, unlike when considering the outputs of $\mathcal{ALG}_\textit{mon}$ directly, we do not have a sufficient smoothness property regarding sums of reduced weights at our disposal that would allow us to bound these quantities. 

Our key idea is as follows. To establish the necessary smoothness properties, we employ the smooth histogram technique directly on sums of reduced weights rather than on the size of the output matching itself. To be consistent with the literature and to illustrate the increment over Biabani et al.'s work, we encapsulate this idea via an alternative definition of lookahead algorithms, denoted {\em refined lookahead algorithms} (see Definition~\ref{def:refined-lookahead} for details), which enables us to incorporate the required smoothness property of sums of reduced weights into the definition. 
We then prove that, similar to lookahead algorithms, refined lookahead algorithms can still be turned into sliding window algorithms with a similar small increase in the space complexity. 
Last, we finish our argument by proving that {\mwm} is a refined lookahead algorithm with an approximation factor of $3+\varepsilon$, which establishes our result.

\paragraph*{Further Related Work}
The sliding window model can be regarded as a streaming insertion-deletion model with highly structured deletions since, for each incoming edge, the oldest edge in the current window is deleted. 
Interestingly, the complexities of \textsf{MM} and \textsf{MWM} in the sliding window model are much closer to those in the insertion-only model, where no deletions are allowed, as opposed to the insertion-deletion model, where arbitrary deletions are allowed. 
In the insertion-only model,  the currently best one-pass algorithm known for \textsf{MM} is the \textsc{Greedy} matching algorithm, which produces a $2$-approximation and uses semi-streaming space $\Ot(n)$. 
It is known that computing a $(1+\ln 2)$-approximation requires strictly more space than $\Ot(n)$ \cite{Kapralov2021SpaceLB}, see also the previous lower bounds \cite{gkk12,Kapralov2013BetterBF}. 
It remains a key open problem to close this gap. Regarding  \textsf{MWM}, a series of works  \cite{Feigenbaum2005OnGP,m05,z12,Epstein2011ImprovedAG,Crouch2014ImprovedSA,ps17,gw19} culminated in the Paz and Schwartzman algorithm, which closes the gap between \textsf{MWM} and \textsf{MM} from an algorithmic perspective in the insertion-only model. 
In the insertion-deletion model, where arbitrary previously inserted edges can be deleted again, it is known that space $\Theta(n^2 / \alpha^3)$ is necessary and sufficient for computing an $\alpha$-approximation to \textsf{MM}, see \cite{as22} for the algorithm and \cite{dk20} for a matching lower bound (see also the previous works \cite{k15,akly16}). 
\textsf{MWM} reduces easily to \textsf{MM} in the insertion-deletion model, by, for example, grouping edges of similar weights into groups and running the \textsf{MM} algorithm a logarithmic number of times in parallel at the expense of only a marginal increase in the approximation factor.

The sliding window model is inspired by the problem of inferring statistics of data occurring within a certain time frame over a continuous stream of data (e.g., maintaining the number of distinct users who have accessed a social media page in the last 24 hours). 
The model was introduced by Datar et al.~\cite{Datar2002MaintainingSS}, and Crouch et al.~\cite{cms13} were the first to consider graph problems in the sliding window model. 
Among others, they showed that testing \textsf{Connectivity} and \textsf{Bipartiteness}, and constructing $(1+\varepsilon)$-sparsifiers can be done in the sliding window model using semi-streaming space. 
Furthermore, as previously mentioned, they also gave the first sliding window algorithms for \textsf{MM} and \textsf{MWM}. 

The smooth histogram technique used in our work was introduced by Braverman and Ostrovsky~\cite{Braverman2007SmoothHF} and can be regarded as an improvement of the exponential histogram technique \cite{Datar2002MaintainingSS} for smooth functions. 
This technique has successfully been applied to a wide range of problems, including the computation of coresets \cite{Wang2019CoresetsFM} and for clustering problems \cite{Braverman2016ClusteringPO}.

\paragraph*{Outline}
We first give notation and a discussion of Paz and Schwartzman's algorithm including its properties in Section~\ref{sec:Preliminaries}. 
The $(2+\varepsilon)$-approximation is presented in Section~\ref{sec:twoapprox}.
The semi-streaming $(3 + \varepsilon)$-approximation via the refined lookahead algorithms is then given in Section~\ref{sec:threeapprox}.
Finally, we conclude with open questions in Section~\ref{sec:conclusion}. 

\section{Preliminaries}
\label{sec:Preliminaries}

In this section, we start with some important notation and a formal description of the improved version of Paz and Schwartzman's algorithm by Ghaffari and Wajc (see \Cref{alg:Streaming}). 
This is followed by some key insights about the algorithm.

Let $S$ be an input stream representing an edge-weighted graph $G = (V, E, w)$ with a weight function $w: E \to \R^+$.
We assume that each edge, including its weight, can be stored in a single word of memory; as such, all our space bounds are in terms of words of memory. 
For any subset of edges $F \subseteq E$, let $w(F) = \sum_{e \in F} w(e)$ be the sum of their weights. 
Then, for any maximum-weight matching in $G$, denoted by $M^*(S)$, we have that $\textsf{MWM}(S) = \w{M^*(S)}$.

\begin{algorithm*}[ht]
 \caption{{\mwm} (Paz and Schwartzman's algorithm with Ghaffari and Wajc's improvements)}
 \label{alg:Streaming}
 \vspace{2mm}
 \textbf{Input:} A stream $S$ of weighted edges  
 \vspace{2mm}
  \hrule
\vspace{2mm}
\textbf{Initialization:}
\begin{algorithmic}[1]
\State \textsf{Stack} $\leftarrow$ an empty stack
\For{every vertex $v \in V$} $\varphi(v) \leftarrow 0$ \EndFor
\label{streaming:initend}
\end{algorithmic}
\vspace{2mm} 
\hrule
\vspace{2mm}
 \textbf{Streaming:}
 \begin{algorithmic}[1]  \setcounter{ALG@line}{\getrefnumber{streaming:initend}}
  \While{a new edge $e = \{u,v\}$ of the stream $S$ is revealed}
   \If{$w(e) < (1 + \varepsilon)\cdot \bigl(\varphi(u) + \varphi(v)\bigr)$} $w'(e) \leftarrow 0$
   \Else
   \State $w'(e) \leftarrow w(e) - \bigl(\varphi(u) + \varphi(v)\bigr)$ \Comment{$w'(e)$ is the \emph{reduced weight of $e$}}
   \State $\varphi(u) \leftarrow \varphi(u) + w'(e);$ $\varphi(v) \leftarrow \varphi(v) + w'(e)$ \Comment{update potentials}
   \State {\sf Stack.Push}($e$)
   \EndIf
   \For{$x \in \{u,v\}$} \Comment{optimizing space}
   \If{$x$ is adjacent to $> \frac{3 \log(1/\varepsilon)}{\varepsilon} + 1$ edges in {\sf Stack}}
   \State Remove the oldest edge adjacent to $x$ from {\sf Stack}
   \EndIf
   \EndFor
  \EndWhile
  \label{streaming:streamend}
 \end{algorithmic}
 \vspace{2mm}
 \hrule
\vspace{2mm} 
 \textbf{Postprocessing:}
 \begin{algorithmic}[1] \setcounter{ALG@line}{\getrefnumber{streaming:streamend}}
  \State Let $\hat{M}$ be an empty matching
  \While{{\sf Stack} is not empty}
  \State $e \leftarrow$ {\sf Stack.Pop}()
  \If {$\hat{M} \cup \{e\}$ is a matching} $\hat{M} \leftarrow \hat{M} \cup \{e\}$
  \EndIf
  \EndWhile
  \State \Return{$\hat{M}$} \Comment{a \textsc{Greedy} matching of the edges in {\sf Stack}}
 \end{algorithmic}

\end{algorithm*}

{\mwm} (\Cref{alg:Streaming}) uses the notions of reduced weights and vertex potentials. 
These are respectively represented by the functions $w'_S: E \to \R^+_0$ and $\varphi_S: V \to \R^+_0$ when the algorithm is executed on a stream $S$. 
The sum of all reduced weights is denoted by $W'_S = \sum_{e \in S} w'_S(e)$. 
For any edge in the stream, its reduced weight is non-negative and is unchanged by the processing of any subsequent edges.
In particular, for a stream $AB$ and any edge $e \in A$ (i.e., the edge $e$ is present in the stream $A$), we have $w'_{A}(e) = w'_{AB}(e) \geq 0$. 
Hence, the sum of the reduced weights is a non-decreasing function, i.e., $W'_A \leq W'_{AB}$. 
The output matching of {\mwm} on stream $S$ is denoted by $\hat{M}(S)$.

Ghaffari and Wajc's analysis of the algorithm reveals the following key observations and results which we later use in our proofs.

\begin{observation}[Ghaffari and Wajc \cite{gw19}] \label{obs:StreamingNumEdges}
 At any moment there are $O\left(\frac{\log(1/\varepsilon)}{\varepsilon}\cdot n\right)$ edges stored in {\sf Stack} during the execution of {\mwm}.
\end{observation}

\begin{proposition}[Ghaffari and Wajc \cite{gw19}]
 \label{lem:WeightBound}
 For any edge $e = \{u,v\}$ in a stream $S$, after the execution of {\mwm}, its weight is bounded as $w(e) \leq (1 + \varepsilon)\cdot \bigl(\varphi_S(u) + \varphi_S(v)\bigr)$.
\end{proposition}

\begin{proposition}[Ghaffari and Wajc \cite{gw19}]
\label{thm:Streaming2apx}
 Let $\varepsilon > 0$ and $S$ be a stream of edges.
 Then, the following inequalities hold: 
 \begin{align*}
  \w{M^*(S)} 
  &\geq W'_S, \\
  \w{\hat{M}(S)} 
  &\geq \frac{1}{1 + 4\varepsilon} \cdot W'_S = \frac{1}{2(1 + 4\varepsilon)} \sum_{v \in V} \varphi_S(v) 
  \geq \frac{1}{2(1 + 4\varepsilon)(1 + \varepsilon)} \cdot \w{M^*(S)}. 
 \end{align*}
\end{proposition}
Note that \cref{thm:Streaming2apx} uses the important fact that $W'_S = \frac{1}{2} \sum_{v \in V} \varphi_S(v)$ as the potential of a vertex $v$ is actually the sum of reduced weights of edges incident to $v$.
Furthermore, its last inequality is due to \cref{lem:WeightBound} since each vertex in a matching is incident to at most one edge. 
Indeed, \Cref{thm:Streaming2apx} shows that {\mwm} is a $(2 + \varepsilon)$-approximation streaming algorithm for \textsf{MWM}, and, by \cref{obs:StreamingNumEdges}, {\mwm} uses space $O(\frac{\log(1/\varepsilon)}{\varepsilon} \cdot n)$ (in words).

\section{\texorpdfstring{$(2 + \varepsilon)$}{(2 + epsilon)}-approximation Sliding Window Algorithm}
\label{sec:twoapprox}
In this section, we give a $(2 + \varepsilon)$-approximation sliding window algorithm for \textsf{MWM} with space $\tilde{O}(\sqrt{nL})$, where $L$ is the length of the sliding window.

\begin{algorithm*}[ht]
  \caption{{\sc MWM Sliding Window Algorithm}} 
  \label{alg:mwmsliding}
  \vspace{2mm}
  \textbf{Input:} A stream $S$ with a sliding window of length $L$ \\
  $\cA$: {\mwm} with sum of reduced weights $W'$ and output matching $\hat{M}$.
  \vspace{2mm} 
  \hrule
  \vspace{2mm}
  \textbf{Initialization:} 
  \begin{algorithmic}[1]
  \State \textsf{Stack} $\leftarrow$ an empty stack
  \State $k \leftarrow 0$ \Comment{Number of blocks}
  \State Parameter $s \leftarrow \Bigl\lfloor  \frac{\sqrt{ n \cdot L \cdot \log{1/\varepsilon} \cdot \log{\sigma}}}{\varepsilon} \Bigr\rfloor$ for $\sigma = \frac{n}{2}\cdot w_{\text{max}}/w_{\text{min}}$.
  \label{abcdefgh}
  \end{algorithmic}
  \vspace{2mm} 
  \hrule
  \vspace{2mm}
  \textbf{Streaming:} 
  \begin{algorithmic}[1] \setcounter{ALG@line}{\getrefnumber{abcdefgh}}
  \While{a new item $e$ of the stream $S$ is revealed}
  \State Feed $e$ to all existing instances of $\cA$
  \State Delete all instances of $\cA$ which have processed more than $L$ edges
  \State \textsf{Stack.Push}($e$)
  \If{$|\textsf{Stack}| \geq s$} \Comment{Create new instances of $\cA$}
  \State $k \leftarrow k + 1$
  \label{ln:StartCreating}
    \State Let $\cI^k_1$ be a new instances of $\cA$
    \State Let $W'_\text{prev} \leftarrow 0, i \leftarrow 1$
    \While{\textsf{Stack} is not empty} \Comment{Process all edges in reverse order of arrival}
      \State $e' \leftarrow \textsf{Stack.Pop}$ and feed $e'$ to $\cI^k_i$
      \If{$W'(\cI^k_i) > (1+\varepsilon)\cdot W'_\text{prev}$} \Comment $W'(\cI^k_i)$ exceeds $(1 + \varepsilon)\cdot W'(\cI^k_{i-1})$
      \label{ln:NewInstance}
        \State Create a new instance $\cI^k_{i+1}$ as a copy of $\cI^k_i$
        \State $W'_\text{prev} \leftarrow  W'(\cI^k_i), i \leftarrow i + 1$
      \EndIf
    \EndWhile
  \EndIf
  \label{ln:EndCreating}
  \If{any instance of $\cA$ exists} 
    \State \textbf{report} output matching of the instance that has processed the most edges
  \Else{} \textbf{report} the maximum-weight matching of the edges in \textsf{Stack}
  \EndIf
  \EndWhile
  \end{algorithmic}
  \end{algorithm*}

For brevity of notation, denote by $\cA$ the Paz and Schwartzman algorithm {\mwm}, which our algorithm (see \Cref{alg:mwmsliding} for a listing) maintains several instances of.
When the current edge $e$ of the stream arrives, the algorithm
feeds $e$ to all existing instances of $\cA$, then deletes any instance that has processed more than $L$ edges, i.e., the ones that could return edges outside the sliding window.
The edge $e$ is subsequently pushed onto {\Stack}.

When {\Stack} has accumulated $s$ edges, \Cref{alg:mwmsliding} uses it to create several instances of $\cA$:
It first creates a new instance $\cI_1$ of $\cA$, then starts to pop the edges from {\Stack}, processing the edges in reverse order of their arrival. 
When an edge $e$ is popped it is fed into the last created instance $\cI_i$ (initially $\cI_1$).
At any given moment, the algorithm stores the sum of reduced weights $W'(\cI_{i-1})$ of the previous instance (initially set to $0$).
If the sum of reduced weights $W'(\cI_i)$ of the latest instance $\cI_i$ exceeds $(1 + \varepsilon)\cdot W'(\cI_{i-1})$, then a new instance $\cI_{i + 1}$ is created as a copy of $\cI_i$. 
This procedure is repeated until {\Stack} is empty again.

After processing edge $e$, the algorithm reports the matching computed by the instance of $\cA$ which has processed the most edges of the current sliding window.
If no instances have been created yet, then it reports an exact solution on the edges stored in {\Stack}.

Overall, \Cref{alg:mwmsliding} maintains multiple runs of $\cA$, each fed with different suffixes of the sliding window. 
It uses {\Stack} to implicitly partition the stream $S$ into blocks $B_1, B_2, \dots$ of $s$ edges each, thus processing it block by block.
Each block $B_j$ is then processed, crucially in reverse order of arrival, feeding each edge into an initially empty instance $\cI^j_1$ of $\cA$.
Then, copies $\cI^j_i$ are created whenever the sum of reduced weights exceed a $(1 + \varepsilon)$ factor of the previous copy.
Once the block $B_j$ has been processed entirely, the subsequent edges of the stream are fed to the instances $\cI^j_1, \cI^j_2, \dots, \cI^j_\ell$ in the natural arrival order.
Note that the algorithm constructs the instances such that $\cI^j_1$ only processes a single edge of the block $B_j$ and $\cI^j_\ell$ processes the entire block.

Intuitively, \Cref{alg:mwmsliding} ensures that, as edges of the block start to fall outside of the sliding window, the oldest remaining instance is still a good approximation of the solution on the entire sliding window, i.e., consecutive runs of $\cA$ are not too different in terms of output.
Moreover, immediately after processing block $B_j$, it holds  that $W'(\cI^j_i) > (1 + \varepsilon)\cdot W'(\cI^j_{i - 1})$ for all $1 < i \leq \ell$.
Therefore, there are only logarithmically many runs of $\cA$ per block.

In the following proofs, we use a notion $S(\cI)$ to denote a substream that is processed by the instance $\cI$ of $\cA$.

\begin{figure}[ht]
 \centering
 \includegraphics[scale=0.82]{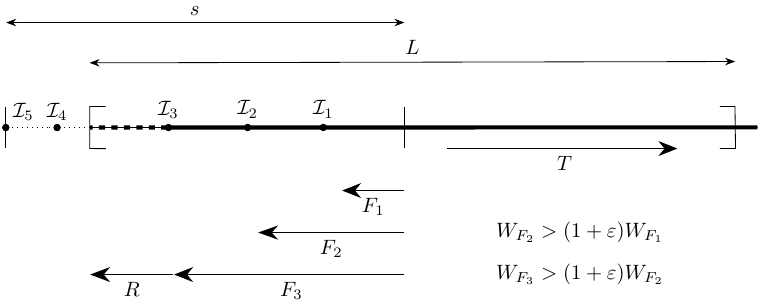}
 \caption{A schematic of a block of the algorithm. The notation here coincides with the notation used in the proof of \cref{thm:TwoApprox}. 
 The window of length $L$ is marked between the square brackets. 
 There are five instances $\cI_1, \dots, \cI_5$ created for the block of length $s$.
 The instance $\cI_{i}$ processed the stream $\overleftharpoon{F_i}T$. 
 Note that $\cI_4$ and $\cI_5$ are already expired, thus they were deleted. 
 The algorithm outputs the result of the instance $\cI_3$ meaning the stream $\overleftharpoon{F_3}T$. 
 The proof of \cref{thm:TwoApprox} will show that the omission of the remainder $\overleftharpoon{R}$ does not compromise the output matching too much.}
 \label{fig:twoapprox_example}
\end{figure}

\begin{theorem}\label{thm:TwoApprox}
There is a deterministic streaming sliding window algorithm for {\sf Maximum-weight Matching} with an approximation factor $2 + \varepsilon$ that uses $O\left(\frac{\sqrt{n \cdot L \cdot \log{1/\varepsilon} \cdot \log{\sigma}}}{\varepsilon}\right)$ words of memory for any $\varepsilon > 0$ and $\sigma = \frac{n}{2}\cdot w_{\text{max}}/w_{\text{min}}$.
\end{theorem}

\begin{proof}
 We will prove that \cref{alg:mwmsliding} satisfies the assertion of the theorem.
 Let $B_j$ be the oldest block of the stream which is still partially contained in the current sliding window $E$, i.e., $E$ contains at least one edge of $B_j$ and no edge of $B_{j - 1}$.
 Let $\cI^j_1,\dots,\cI^j_\ell$ be the instances created during the processing of block $B_j$. 
 Note that each instance $\cI^j_i$ processes the edges of $B_j$ in reverse order.
 Thus, we consider $B_j$ as a stream of edges ordered in reverse to the order in which they arrived.
 For clarity, we denote this as $\overleftharpoon{B_j}$ and similarly for all relevant substreams of $\overleftharpoon{B_j}$. 
 Let $\overleftharpoon{F_i}$ be the substream of $\overleftharpoon{B_j}$ fed into the instance $\cI^j_i$, i.e., $\overleftharpoon{F_i} = S(\cI^j_i) \cap \overleftharpoon{B_j}$.
 Note that $\overleftharpoon{F_1} \subseteq \dots \subseteq \overleftharpoon{F_\ell} = \overleftharpoon{B_j}$.
 
 \subparagraph{Approximation.}
 Let $T$ be the stream of edges that arrive after the stream $\overleftharpoon{B_j}$, i.e., $E \subseteq \overleftharpoon{B_j}T$.
 First suppose that $E = \overleftharpoon{B_j}T$.
 Then, \cref{alg:mwmsliding} returns the matching computed by the oldest instance $\cI^j_\ell$ which has processed the whole stream $\overleftharpoon{F_\ell} T$, i.e., all edges of $E$ (as $\overleftharpoon{F_\ell} = \overleftharpoon{B_j}$).
 Thus, it returns a $(2 + \varepsilon)$-approximation of the optimal solution.
 
 Now, suppose that $E \subset \overleftharpoon{B_j}T$.
 Let $\overleftharpoon{F_i}T \subseteq E \subset \overleftharpoon{F_{i+1}}T$.
 Note that such an $i$ exists as $E \cap \overleftharpoon{B_j} \neq \emptyset$ and $|\overleftharpoon{F_1}| = 1$.
 \cref{alg:mwmsliding} returns a matching computed by the instance $\cI^j_i$ that processed the stream $S(\cI^j_i) = \overleftharpoon{F}T$ for $\overleftharpoon{F} = \overleftharpoon{F_i}$. 
 Let $\overleftharpoon{R}$ be the substream of $\overleftharpoon{F_{i+1}} \setminus \overleftharpoon{F}$ such that $E$ contains exactly the edges of the stream $\overleftharpoon{F}\overleftharpoon{R}T$.
 See Figure~\ref{fig:twoapprox_example}, for an illustration of the substreams processed by various instances $\cI_i$.
 
 Since $E \subset \overleftharpoon{F_{i+1}}T$, it holds by construction of \cref{alg:mwmsliding} that $W'_{\overleftharpoon{F}\overleftharpoon{R}} \leq (1 + \varepsilon)\cdot  W'_{\overleftharpoon{F}}$, where $W'_{\overleftharpoon{F}\overleftharpoon{R}}$ and $W'_{\overleftharpoon{F}}$ are the sums of reduced weights computed by $\cA$ on streams $\overleftharpoon{F}\overleftharpoon{R}$ and $\overleftharpoon{F}$, respectively.
 For any vertex $v$, let $\varDelta(v) := \varphi_{\overleftharpoon{F}\overleftharpoon{R}}(v) - \varphi_{\overleftharpoon{F}}(v)$. 
 Recall that $\varphi$ is an increasing function by the construction of the algorithm, thus  $\varDelta(v) \ge 0$.
 Then, by the proportionality between the sum of reduced weights and the sum of potentials ($\sum_e w'(e) = 2\sum_v \varphi(v)$, see \cref{thm:Streaming2apx}), we have the following upper bound:
 \[
 \sum_{v \in V} \varDelta(v) = \sum_{v \in V} \varphi_{\overleftharpoon{F}\overleftharpoon{R}}(v) - \varphi_{\overleftharpoon{F}}(v) \le \sum_{v \in V} (1 + \varepsilon)\cdot \varphi_{\overleftharpoon{F}}(v) - \varphi_{\overleftharpoon{F}}(v) = \varepsilon \cdot \sum_{v \in V}  \varphi_{\overleftharpoon{F}}(v) \ . 
 \]
 We now claim that if we assign, for every $v \in V$, a weight $c(v) := (1+\varepsilon)\cdot \bigl(\varphi_{\overleftharpoon{F}T}(v) + \varDelta(v)\bigr)$, then we have a valid (weighted) vertex cover in the graph consisting of all edges in $\overleftharpoon{F}\overleftharpoon{R}T$, i.e., for each edge $e = \{u,v\} \in \overleftharpoon{F}\overleftharpoon{R}T$, it holds that $c(e) := c(u) + c(v)  \ge w(e)$.
 Consider two cases.
 If $e \in \overleftharpoon{F}T$, then we have $c(e) \ge (1+\varepsilon)\cdot \bigl(\varphi_{\overleftharpoon{F}T}(v) + \varphi_{\overleftharpoon{F}T}(u)\bigr) \ge w(e)$.
 Otherwise, $e \in \overleftharpoon{R}$ and we have
 \begin{align*}
 c(e) &= (1+ \varepsilon)\cdot \bigl(\varphi_{\overleftharpoon{F}T}(v) +  \varphi_{\overleftharpoon{F}T}(u) + \varphi_{\overleftharpoon{F}\overleftharpoon{R}}(v) - \varphi_{\overleftharpoon{F}}(v) + \varphi_{\overleftharpoon{F}\overleftharpoon{R}}(u) - \varphi_{\overleftharpoon{F}}(u)\bigr)   \\
 &\ge (1 + \varepsilon)\cdot \bigl(\varphi_{\overleftharpoon{F}\overleftharpoon{R}}(v) + \varphi_{\overleftharpoon{F}\overleftharpoon{R}}(u) \bigr)  \tag{$\overleftharpoon{F}$  is a substream of  $\overleftharpoon{F}T$}\\ 
 &\ge  w(e) \ . \tag{by \cref{lem:WeightBound}}
 \end{align*}
Thus, we get a valid vertex cover as required.
Now, we can use this to show that the returned matching $\hat{M}(\overleftharpoon{F}T)$ computed by $\cA$ on the stream $\overleftharpoon{F}T$ is a $(2 + \varepsilon)$-approximation of the maximum weighted matching $M^*$ of the sliding window $E$:
 \begin{align*}
 w(M^*) & = \sum_{e \in M^*} w(e) \le \sum_{v \in V}c(v) \tag{each vertex is incident to at most one edge in $M^*$} \\
  &= (1+\varepsilon)  \sum_{v \in V} \bigl(\varphi_{\overleftharpoon{F}T}(v) + \varDelta(v)\bigr) \\
       &\le (1+\varepsilon) \sum_{v \in V} \bigl(\varphi_{\overleftharpoon{F}T}(v) + \varepsilon \varphi_{\overleftharpoon{F}}(v)\bigr) \\
       & \le (1+ \varepsilon)^2 \sum_{v \in V} \varphi_{\overleftharpoon{F}T}(v) \tag{since $\varphi$ is monotonic}\\
       & \le (1+ 3\varepsilon) \sum_{v \in V} \varphi_{\overleftharpoon{F}T}(v) \tag{since $\varepsilon^2 < \varepsilon$ for $0 < \varepsilon < 1$} \\
       & \le 2 (1 + 3 \varepsilon)(1 + 4\varepsilon) \cdot \w{\hat{M}(\overleftharpoon{F}T)} \tag{by \cref{thm:Streaming2apx}}\\
       & \le (2 + 38 \varepsilon) \cdot  \w{\hat{M}(\overleftharpoon{F}T)} \ .
 \end{align*}
 \subparagraph{Space.}
 The sliding window $E$ can be covered by $O\left(\frac{L}{s}\right)$ many blocks, as $s$ is the block size.
 First, we bound the number of instances $\ell$ created for a block $\overleftharpoon{B_j}$.
 Recall that edges from the block $B_j$ processed by the instance $\cI^j_i$ are the edges in $\overleftharpoon{F_i}$, and $\overleftharpoon{F_1} \subseteq \dots \subseteq \overleftharpoon{F_\ell} = \overleftharpoon{B_j}$. 
 Furthermore, $\overleftharpoon{F_1} = \{e'\}, W'_{\overleftharpoon{F_1}} = w(e') \geq w_\text{min}$, and $W'_{\overleftharpoon{F_{i + 1}}} > (1 + \varepsilon) \cdot W'_{\overleftharpoon{F_i}}$ for all $i < \ell$.
 Thus,
 \[
   (1 + \varepsilon)^{\ell - 1} \cdot w_\text{min} \leq  (1 + \varepsilon)^{\ell - 1} \cdot W'_{\overleftharpoon{F_1}} < W'_{\overleftharpoon{F_\ell}} \leq \w{\hat{M}(\overleftharpoon{F_\ell})} \leq \frac{n}{2}\cdot w_\text{max}.
 \]
 By rearranging, we get $\ell = O(\log_{1 + \varepsilon} \sigma) = O\left(\frac{1}{\varepsilon} \cdot \log \sigma \right)$.
 By \cref{obs:StreamingNumEdges}, each instance of $\cA$ stores $O\left(\frac{n \log(1/\varepsilon)}{\varepsilon}\right)$ edges.
 Thus, at any moment, all existing instances of $\cA$ store $O\left(\frac{L}{s} \cdot \frac{\log \sigma}{\varepsilon} \cdot \frac{n \log(1/\varepsilon)}{\varepsilon}\right)$ many edges.
 
 Note that we additionally need to store the edges of at most one block (stored in {\Stack}), i.e., at most $s$ edges.
 Overall, we need to store at most  $O\left(\frac{L}{s} \cdot \frac{\log  \sigma}{\varepsilon} \cdot \frac{n \log(1/\varepsilon)}{\varepsilon} + s\right)$
edges.
Setting $s$ to $\left\lfloor  \frac{\sqrt{  n \cdot L \cdot \log{1/\varepsilon} \cdot \log{\sigma}}}{\varepsilon} \right\rfloor$ gives us the final space bound in words of memory.
 \end{proof}
 
\begin{remark*}
 Assuming that $\varepsilon$ is constant and that $\sigma$ is polynomial in $n$, we obtain an algorithm that uses $\tilde{O}(\sqrt{nL})$ space. This is $o(n^2)$ space as long as $L = \tilde{o}(n^3)$. If, additionally, the input graph of each window is simple, we have that $L = O(n^2)$ (a simple graph always has $O(n^2)$ edges) and a space bound of $O\left( n \sqrt{n} \cdot \frac{   \sqrt{  \log{1/\varepsilon} \cdot \log{ \sigma} }  }{ \varepsilon}  \right)$, which simplifies to $\Tilde{O}\left(n \sqrt{n}\right)$.
\end{remark*}

We can easily adapt the algorithm to the (unweighted) \textsf{MM} problem. 
More specifically, the Paz-Schwartzman algorithm becomes the \textsc{Greedy} matching algorithm, while the sum of reduced weights simply becomes the size of the \textsc{Greedy} matching obtained. 
While the approximation factor remains $2 +  \varepsilon$, the matchings of the instances now store $O(n)$ edges instead of $O\left(\frac{n \log(1/\varepsilon)}{\varepsilon}\right)$. Also, $\sigma = \frac{n}{2}$. 
Then, by setting $s$ to $\left\lfloor  \sqrt{ \frac{ n \cdot L \cdot  \log{n}}{\varepsilon}} \right\rfloor$, we obtain a better memory bound for the algorithm.
 This adaptation yields the following result:
 
 \begin{theorem}
 There is a deterministic streaming sliding window algorithm for \textsf{MM}  with an approximation factor of $(2 + \varepsilon)$ that uses $O\left(\sqrt{ \frac{ n \cdot L \cdot \log{n} } {\varepsilon}  }  \right)$ words of memory for any $\varepsilon > 0$.
 \end{theorem}

\section{\texorpdfstring{$(3 + \varepsilon)$}{(3 + epsilon)}-approximation Sliding Window Algorithm}
\label{sec:threeapprox}

In this section, we give a $(3 + \varepsilon)$-approximation semi-streaming sliding window algorithm by applying the smooth histogram technique~\cite{Braverman2007SmoothHF} in a similar manner as Biabani et al.~\cite{Biabani21_WeightedMatching}.
We start with our definition of a {\em refined lookahead algorithm} which we use to describe a sliding window algorithm. 
Then, we show that {\mwm} is a refined lookahead algorithm; thus, obtaining the sliding window algorithm for \textsf{MWM}.




\begin{restatable}[$(f, \alpha_1, \alpha_2, \beta)$-refined lookahead algorithm]{definition}{defrefinedlookahead} \label{def:refined-lookahead}
Let $\beta \in (0, 1)$, $\alpha_1, \alpha_2 \geq 1$ and, for a ground set $X$, let ${f: 2^X \to \R^+}$ be a non-decreasing function. 
We say a streaming algorithm $\alg$ with two outputs $O_1, O_2$ is a \emph{$(f, \alpha_1, \alpha_2, \beta)$-refined lookahead} algorithm if the following holds for any stream $S$ of items of the set $X$:
\begin{enumerate}
 \item $O_1(S) \leq f(S) \leq \alpha_1 \cdot O_1(S)$, i.e., the first output is an $\alpha_1$-approximation of $f$. \label{lookahead:prop1}
 \item For any partitioning of $S$ into three disjoint sub-streams $A$, $B$, and $C$ with $O_1(B) \geq (1 - \beta) \cdot O_1(AB)$, we have $O_2(BC) \leq f(ABC) \leq \alpha_2 \cdot O_2(BC)$, i.e., if the first output on the substream $AB$ is similar to the first output on the substream $B$ then the second output on the substream $BC$ is an $\alpha_2$-approximation of $f$ on the whole stream $S = ABC$. \label{lookahead:prop2}
\end{enumerate}

\end{restatable}

Observe that if $O_1 = O_2$ and $\alpha_1 = \alpha_2 = \alpha$ then we retrieve the standard definition of a $(f,\alpha, \beta)$-lookahead algorithm as given by Biabani et al.\ (see \cref{def:lookahead}).
Our refined lookahead algorithm is also similarly turned into a sliding window algorithm.
In essence, the algorithm simulates runs of a traditional streaming algorithm on suffixes of the current sliding window.
It maintains runs on suffixes such that the value of $O_1$ of any two consecutive runs are not too different, while the value of $O_1$ of any non-consecutive runs are sufficiently different so as to ensure that at most a logarithmic number of runs are required at any point of time.
The second output $O_2$ is a solution which, given the smoothness assumptions of the runs, is always guaranteed to be an $\alpha_2$-approximation of the next oldest run. 
Details of the algorithm and the proof of the following theorem are provided in \Cref{app:reflook}.

%

\begin{restatable}{theorem}{thmsliding} \label{thm:sliding}
  Let $0 < \beta < 1$ and $\alpha_1,\alpha_2 \geq 1$, $S$ be a stream of items from a set $X$, and $f: 2^X \to \R^+$ be a non-decreasing function.
  Suppose there exists a $(f, \alpha_1, \alpha_2, \beta)$-refined lookahead algorithm that uses at most $s$ words of memory. 
  Then, there is a sliding window algorithm that maintains an $\alpha_2$-approximation of $f$ using $O\bigl(\frac{1}{\beta} \cdot s  \log (\alpha_1\sigma)\bigr)$ words of memory for $\sigma = f(S) / f_{\min}$ where $ f_{\min} = \min \{f(e) : e \in S\}$.
 \end{restatable}

We will now apply \Cref{def:refined-lookahead} to algorithm {\mwm}.
To this end, we consider the first output $O_1$ as the sum of reduced weights $W'_S$, the second output $O_2$ as the weight of the returned matching $\w{\hat{M}(S)}$, and function $f$ as the weight of a maximum-weight matching $\textsf{MWM}(S)$. 
In fact, we prove in \Cref{thm:LookaheadMatching} that this indeed yields a $\bigl(\textsf{MWM}, (2+2\varepsilon), (3 + 20\varepsilon), \beta\bigr)$-refined lookahead algorithm.
Hence, the algorithm given by Theorem~\ref{thm:sliding} with {\mwm} is a $(3+\varepsilon)$-approximation semi-streaming sliding window algorithms for \textsf{MWM}.

\begin{theorem}
  \label{thm:LookaheadMatching}
   Let $0 < \varepsilon \leq \frac{1}{10}$ and $0 < \beta \leq \frac{\varepsilon}{9}$.
   The algorithm {\mwm} is a $\bigl(\textsf{MWM}, {(2 + 2\varepsilon)}, \\ {(3 + 20\varepsilon)}, \beta\bigr)$-refined lookahead algorithm.
  \end{theorem}

To prove \cref{thm:LookaheadMatching}, we follow the approach of Biabani et al.~\cite{Biabani21_WeightedMatching}.
Let an input stream $S$ be partitioned into three substreams $ABC$.
They split the maximum matching of the stream $M^* = M^*(ABC)$ into two parts $M^*_{AB}$ and $M^*_{C}$ where $M^*_{AB} := M^* \cap AB$ is the restriction of $M^*$ to the edges in $AB$, analogously for the substream $C$. 
Biabani et al.\ then bound the weights of these two parts separately.
To this end, they use the notion of a \emph{critical subgraph}.

\begin{definition}[Critical Subgraph~\cite{Biabani21_WeightedMatching}]\label{def:critsub}
Consider a graph $G$ specified by a stream $S$ of edges.
Let $A, B, C$ be disjoint substreams of $S$ such that $S = ABC$. 
Then, the \emph{critical subgraph of $G$} with respect to the maximum matching $M^*(ABC)$ and the substreams $A, B, C$ is the subgraph $H = (V_H, E_H)$ such that
\begin{itemize}
 \item $E_H := \{e \in B \mid e \text{ is adjacent to two edges in  } M^*_{C}\}$. 
 \item $V_H := V(E_H)$, i.e., $V_H$ is the set of endpoints of the edges in $E_H$. 
\end{itemize}
\end{definition}

Biabani et al.\ use the critical subgraph to bound the weights of $M^*_{AB}$ and $M^*_{C}$ in terms of the weight of the matching returned by the algorithm $\w{\hat{M}}$ (Lemmas 13 and 14 in their work~\cite{Biabani21_WeightedMatching}).
 In our analysis, in particular, in  \Cref{lem:ABOptBound,lem:COptBound}, we use the same ideas to bound the weights of $M^*_{AB}$ and $M^*_{C}$ in terms of sums of reduced weights computed by the algorithm instead.

Before stating and proving \Cref{lem:ABOptBound,lem:COptBound}, we present the following auxiliary lemma already proved by Biabani et al.\ in the exact formulation as we need it.
We highlight that their proof holds for any run of {\mwm} on an arbitrary stream.

\begin{lemma}[Biabani et al.~\cite{Biabani21_WeightedMatching}, Lemma 15]
  \label{lem:PotentialBound}
  For any stream $AB$,
   \[
  (1 + \varepsilon)\cdot \sum_{v \in V_H} \varphi_{AB}(v) \geq \sum_{e \in E_H} w'_B(e)  \ .
   \]
\end{lemma}

For the statement of the next auxiliary lemma, we need the following notion.
Let $S$ be a stream of edges.
For an edge $e \in S$, we define the set $P_S(e)$ as the set of edges incident to $e$ (including $e$) arriving no later than $e$, i.e, $ P_S(e) = \{e' \in S \mid e' \cap e \neq \emptyset, t_{e'} \leq t_e \}$, where, for any edge $f$, $t_{f}$ is the arrival time of edge $f$.
Biabani et al.~\cite{Biabani21_WeightedMatching} showed that the weight of any edge $e$ can be bounded by the sum of the reduced weights of the edges in $P_S(e)$ (up to a $(1 + \varepsilon)$ factor).

\begin{lemma}[Biabani et al.~\cite{Biabani21_WeightedMatching}, Lemma 5]
 \label{lem:ReducedWeightsBound}
 For each edge $e \in S$,
 \[
  w(e) \leq (1 + \varepsilon) \sum_{e' \in P_S(e)} w'_S(e).  
 \]

\end{lemma}

With that, we can finally prove our analogous lemmas of Biabani et al.'s Lemmas 13 and 14~\cite{Biabani21_WeightedMatching} which bound $\w{M^*_{AB}}$ and $\w{M^*_{C}}$, respectively.

\begin{restatable}[Analogue of Lemma 13,~\cite{Biabani21_WeightedMatching}]{lemma}{lemABOptBound}
\label{lem:ABOptBound}
 For any stream $ABC$,
 \[
\w{M^*_{AB}} \leq 2(1 + \varepsilon)\cdot W'_{AB} - \sum_{e \in E_H}w'_{B}(e).  
 \]
\end{restatable}

\begin{proof}
 By definition, we have $\w{M^*_{AB}} =  \sum_{e \in M^*_{AB}} w(e)$.  
 Let $e = \{u,v\} \in M^*_{AB}$. 
 Note that the vertices $u$ and $v$ are not in $V_H$. 
 Thus, we can bound the sum as follows.
 
 \begin{align*}
  \w{M^*_{AB}} & \le (1 + \varepsilon) \sum_{v \in V \setminus V_H} \varphi_{AB}(v) & \text{by \cref{lem:WeightBound}}\\ 
  & = (1 + \varepsilon) \left(\sum_{v \in V} \varphi_{AB}(v) - \sum_{v \in V_H} \varphi_{AB}(v)\right)
 \end{align*}
 
 \noindent By \cref{thm:Streaming2apx} and by \cref{lem:PotentialBound}, we have
 \[
   \sum_{v \in V} \varphi_{AB}(v) = 2 W'_{AB}  \text{~~~~~and~~~~~}
  (1 + \varepsilon) \sum_{v \in V_H} \varphi_{AB}(v) \geq \sum_{e \in E_H} w'_B(e).
 \]

\noindent Thus, we can conclude that
 
 \[
 \w{M^*_{AB}} \le 2(1+ \varepsilon) \cdot W'_{AB} - \sum_{e \in E_H} w'_B(e). \qedhere
 \]
 
\end{proof}

\begin{restatable}[Analogue of Lemma 14, \cite{Biabani21_WeightedMatching}]{lemma}{lemCOptBound}
 \label{lem:COptBound}
 For any stream $ABC$,
 \[
\w{M^*_{C}} \leq 2(1 + \varepsilon) \cdot W'_{BC} - (1 + \varepsilon)  \sum_{e \in B \setminus E_H} w'_{B}(e).  
 \]

\end{restatable}

\begin{proof}
  First, when considering a run of the algorithm on $BC$, by \cref{lem:ReducedWeightsBound}, we obtain
  \begin{align*}
   \w{M^*_{C}} & = \sum_{e \in M^*_C} w(e) \leq (1 + \varepsilon) \sum_{e \in M^*_{C}} \sum_{e' \in P(e)} w'_{BC}(e') \ .
  \end{align*}
  Observe that any edge $e \in BC$ is incident to at most two edges of $M^*_C$, and the edges of $B \setminus E_H$ are incident to at most one edge of $M^*_C$. Hence, we can rewrite the previous double sum as follows: 
   
   \begin{align*}
    \sum_{e \in M^*_{C}} \sum_{e' \in P(e)} w'_{BC}(e')
    & \leq 2 \cdot \sum_{e \in BC} w'_{BC}(e) - \sum_{e \in B \setminus E_H} w'_{BC}(e) \\
    & = 2 \cdot W'_{BC} - \sum_{e \in B \setminus E_H} w'_{BC}(e) \ ,
   \end{align*}
  which implies the result.
  \end{proof}

Now, we are ready to prove \cref{thm:LookaheadMatching}, i.e., {\mwm} is a $\bigl(\textsf{MWM}, (2 + 2\varepsilon), (3 + 20\varepsilon), \beta\bigr)$-refined lookahead algorithm for suitable parameters $\varepsilon$ and $\beta$.
\begin{proof}[Proof of \cref{thm:LookaheadMatching}]
 We recall that we consider a version of {\mwm} such that the first output is the sum of reduced weight $W'$ and the second output is the weight of the computed matching $\w{\hat{M}}$.
 First, by \cref{thm:Streaming2apx}, we get that for any stream $S$ it holds that $W'_S \leq \w{M^*(S)} \leq 2(1 + \varepsilon) \cdot W'_S$.
 Thus, it remains to prove that for any stream $ABC$, given that ${W'_{B} \geq (1 - \beta) \cdot W'_{AB}}$, the maximum matching ${M^* = M^*(ABC)}$ is such that ${\w{M^*} \leq (3 + 20\varepsilon) \cdot \w{\hat{M}(BC)}}$.
 \begin{align*}
  \w{M^*} &\leq 2(1 + \varepsilon)\cdot W'_{AB} + 2(1 + \varepsilon) \cdot W'_{BC} - W'_B & \text{by \cref{lem:ABOptBound,lem:COptBound}} \\
  & \leq \frac{2(1 + \varepsilon)}{1 - \beta}\cdot W'_{B} + 2(1 + \varepsilon) \cdot W'_{BC} - W'_{B} & \text{by ${W'_{B} \geq (1 - \beta) \cdot W'_{AB}}$} \\
  & \leq (1 + 3\varepsilon) \cdot W'_{B} + 2(1 + \varepsilon) \cdot W'_{BC} & \text{since }\beta \leq \frac{\varepsilon}{9} \\
  & \leq (3 + 5\varepsilon)\cdot W'_{BC} & \text{by $W'$ being non-decreasing} \\
  & \leq (3 + 5\varepsilon)(1 + 4\varepsilon)\cdot \w{\hat{M}(BC)}  & \text{by \cref{thm:Streaming2apx}} \\
  & \leq (3 + 20\varepsilon)\cdot \w{\hat{M}(BC)} & \text{since } \varepsilon \leq \frac{1}{10} & \qedhere
 \end{align*} 

\end{proof}

\noindent \cref{thm:sliding,thm:LookaheadMatching} together then imply our main result.

\newcounter{thmsaved} 
\setcounter{thmsaved}{\value{theorem}}

\setcounter{theorem}{\value{counterMainThm}}
\addtocounter{theorem}{-1}

\begin{theorem}
 There is a deterministic streaming sliding window algorithm for \textsf{Maximum-weight Matching} with an approximation factor $3+\varepsilon$ that uses $O\left(\frac{\log(1/\varepsilon)}{\varepsilon^2} \cdot n \log \sigma\right)$ words of memory, for any $0 < \varepsilon \le 0.1$ and $\sigma = \frac{n}{2}\cdot w_{\text{max}}/w_{\text{min}}$.
\end{theorem}

\begin{remark*}
  Our $(3+\varepsilon)$-approximation algorithm for \textsf{MWM} yields the $(3+\varepsilon)$-approximation algorithm for \textsf{MM} by Crouch et al.~\cite{cms13} when {\mwm} is replaced with the \textsc{Greedy} matching algorithm (the sum of reduced weights becomes the size of the matching).
  The hard instance of their algorithm also holds for our algorithm.
\end{remark*}

\setcounter{theorem}{\value{thmsaved}}


\section{Conclusion}
\label{sec:conclusion}
In this paper, we gave two algorithms for \textsf{MWM} in the sliding window model. Our first algorithm has an approximation factor of $2+\varepsilon$ and uses space $\Ot(\sqrt{n L})$, and our second algorithm has an approximation factor of $3+\varepsilon$ and uses semi-streaming space. The approximation factor of our semi-streaming algorithm matches the approximation factor of the best semi-streaming sliding window algorithm known for (unweighted) \textsf{MM} \cite{cms13}. 

Regarding the semi-streaming space regime, since further improvements in the approximation factor would imply improvements for (unweighted) \textsf{MM}, the most natural direction for future research is to make further progress on the unweighted version of the problem first. Is there a $2.99$-approximation semi-streaming space sliding window algorithm for \textsf{MM}? 

While the known lower bounds for \textsf{MM} for one-pass streaming algorithms in the insertion-only model also apply to the sliding window model, no stronger lower bounds for the sliding window model are known. Can we prove a lower bound on the approximation factor of sliding window algorithms for \textsf{MM} that use semi-streaming space and are stronger than what is currently known for the insertion-only model, i.e., stronger than $1+\ln(2)$ \cite{Kapralov2021SpaceLB}? 



\bibliography{adkn23} 

\appendix

\section{Hard Instance for Paz and Schwartzman's Algorithm} \label{sec:lb}
In this section, we show that the Paz and Schwartzman's algorithm and its monotonic version are no better than  $(\textsf{MWM}, 3.5, \beta)$-lookahead algorithms. The definition of a lookahead algorithm given by Biabani et al. (Definition~\ref{def:lookahead}) together with the Paz and Schwartzman's algorithm thus cannot be used to improve upon the approximation factor of $3.5$.

Recall that a lookahead algorithm relies on the smoothness of the algorithm's output.
More formally, an $(f,\alpha,\beta)$-lookahead algorithm $\alg$ satisfies the condition that for any stream $ABC$, if $\alg(B) \geq (1 - \beta)\cdot\alg(AB)$ then $f(ABC) \leq \alpha\cdot \alg(BC)$  (see \cref{def:lookahead}).
In other words, if the algorithm $\alg$ outputs similar results on the streams $B$ and $AB$ then the algorithm's output on $BC$ is required to be an $\alpha$-approximation of the objective value  $f(ABC)$ of the whole stream $ABC$.

We will present a graph $G$ whose edges are divided into three substreams $A,B$ and $C$ such that {\mwm} outputs matchings of the same weight on substreams $AB$ and $B$, while the outputted matching on substream $BC$ is roughly a $3.5$-approximation of a maximum-weight matching of the entire stream $ABC$. 
The graph $G$ is such that even if we modified {\mwm} to return maximum-weight matchings among the edges stored in \textsf{Stack} then the same properties still hold. 
Thus, the hard instance is also hard for the monotonic version of the algorithm. 
The graph $G$ is depicted in \cref{fig:lookahead_example}. 

\begin{figure}[ht]
 \centering
 \includegraphics{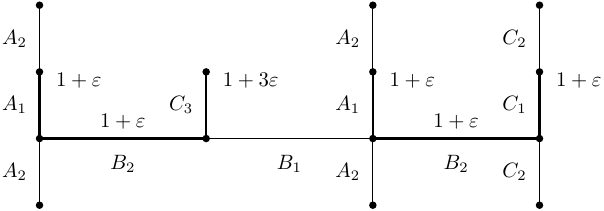}
 \caption{The edges of the graph $G$ are divided into substreams $A,B$ and $C$. The order of the edges within the substreams is indicated by subscripts (the order of the edges with the same subscript is not important). The thin edges have unit weight and the thick edges have the indicated larger weights.}
 \label{fig:lookahead_example}
\end{figure}

\paragraph*{Matchings computed on $AB$ and $B$. } First, we analyze {\mwm} separately on the substreams $A$ and $B$.
See \cref{fig:lookahead_example_AB} for the values of the reduced weights and potentials computed by the algorithm.

\begin{figure}[ht]
 \centering
 \includegraphics{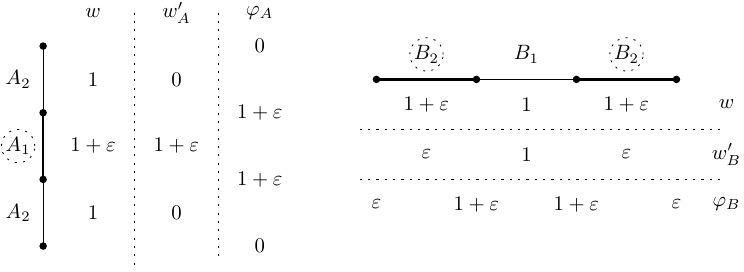}
 \caption{Reduced weights and potentials computed by {\mwm} when run separately on substreams $A$ and $B$. Recall that substream $A$ consists of two paths, while only one of them is depicted here. The edges outputted by the runs of the algorithm are marked by dotted circles.}
 \label{fig:lookahead_example_AB}
\end{figure}

Observe that the substream $A$ consists of two disjoint paths of length three. While only one of them is shown in \cref{fig:lookahead_example_AB}, the algorithm computes the same reduced weights and potentials for both paths.

We now analyze the execution of the algorithm on substream $AB$. To this end, consider the moment when the substream $A$ has been fully processed and substream $B$ begins.
Observe that each edge of $B$ is now incident to a single vertex with potential $1 + \varepsilon$.
Thus, by the construction of the algorithm, none of the edges of $B$ are pushed onto {\sf Stack}. These edges therefore have reduced weights zero and cannot be outputted by the algorithm. Furthermore, when run on $AB$, the algorithm outputs the two edges in $A_1$, i.e., $\w{\hat{M}(AB)} = 2 + 2\varepsilon$, which are the only two edges pushed onto \textsf{Stack}.

As established in Figure~\ref{fig:lookahead_example_AB}, when the algorithm runs only on the substream $B$, it outputs the two edges in $B_2$, i.e., $\w{\hat{M}(B)} = 2 + 2\varepsilon$. Hence, we have that $\w{\hat{M}(B)} = \w{\hat{M}(AB)}$. It follows that the stream $ABC$ satisfies the condition $\w{\hat{M}(B)} \geq (1 - \beta)\cdot \w{\hat{M}(AB)}$, for any value of $\beta \ge 0$, as required by the definition of a lookahead algorithm.

\paragraph*{Matching computed on $BC$.}
Now, we analyze the execution of the algorithm on the substream $BC$.
At the time when the substream $C$ begins, the reduced weights of edges in $B$ and the current potentials of the incident vertices are the same as when the algorithm is run only on the substream $B$ -- see \cref{fig:lookahead_example_AB} for these values.
See \cref{fig:lookahead_example_BC}, for the reduced weights of the edges in $C$ when we run the algorithm on the substream $BC$.

\begin{figure}[ht]
 \centering
 \includegraphics{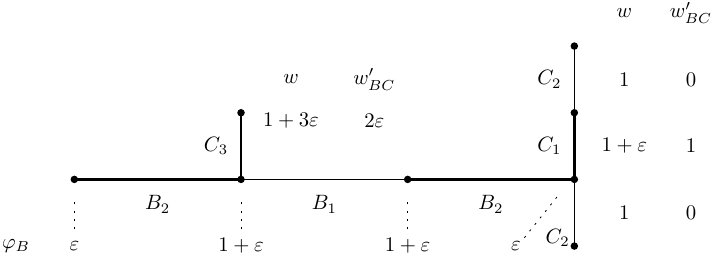}
 \caption{The reduced weights of the edges in $C$ after the execution on the substream $BC$ and the potentials of the vertices incident to the edges in $B$ at the time when the substream $B$ is processed.}
 \label{fig:lookahead_example_BC}
\end{figure}

By the end of the execution, only the two edges in $C_1$ and $C_3$ are pushed onto {\sf Stack} since the edges in $C_2$ have reduced weights zero.
The algorithm {\mwm} outputs a greedy matching of the edges pushed onto {\sf Stack} (in the reverse order they arrived).
In particular, it outputs the edges in $C_1$ and $C_3$ and they block all edges in $B$.
Hence, $\w{\hat{M}(BC)} = 2 + 4\varepsilon$. Observe further that these edges constitute a maximum-weight matching among the edges pushed onto \textsf{Stack}. 

\paragraph*{Maximum-weight Matching and Approximation Factor.}
First, observe that a maximum-weight matching in $G$ consists of all the edges that have an endpoint of degree $1$ (the edges in $A_2, C_2$, and $C_3$) and is thus of weight $7 + 3\varepsilon$.
Since $\w{\hat{M}(BC)} = 2 + 4\varepsilon$, we conclude that it is not possible for {\mwm} to yield a $(\textsf{MWM},3.5 - \varDelta,\beta)$-lookahead algorithm, for any constant $\varDelta > 0$ and suitable parameter $\beta$.

\section{More on Refined Lookahead Algorithms}\label{app:reflook}
In this section, we will prove Theorem~\ref{thm:sliding}. To this end, for convenience, we restate the definition of refined lookahead algorithms first.

\defrefinedlookahead*

\begin{algorithm*}[ht]
\caption{{\sc Lookahead Sliding Window Algorithm}} 
\label{alg:sliding}
\vspace{2mm}
\textbf{Input:} A stream $S$ with a sliding window of length $L$ \\
$\cA$: a $(f, \alpha_1,\alpha_2,\beta)$-refined lookahead algorithm with outputs $O_1$ and $O_2$
\vspace{2mm} 
\hrule
\vspace{2mm}
\textbf{Initialization:} 
\begin{algorithmic}[1]
\State Let $k \leftarrow 0$ be the number of instances
\label{abcdefg}
\end{algorithmic}
\vspace{2mm} 
\hrule
\vspace{2mm}
\textbf{Streaming:} 
\begin{algorithmic}[1] \setcounter{ALG@line}{\getrefnumber{abcdefg}}
\While{a new item $e$ of the stream $S$ is revealed}
\State Create an instance $\cI_{k + 1}$ of $\cA$
\State Feed $e$ into all existing instances $\cI_1,\dots,\cI_{k + 1}$
\State $i \leftarrow 1$
\While{$i < k$} \Comment{Deleting instances with similar value of $O_1$}
\State Let $j > i$ be the largest index for which $O_1(\cI_j) \geq (1-\beta)\cdot O_1(\cI_i)$
\If{no such $j$ exists} $j \leftarrow i + 1$ \EndIf
\State Delete instances $\cI_r$ for each $i < r < j$ 
\label{ln:DeleteSimilarInstances}
\State $i \leftarrow j$
\EndWhile
\State Let $\cI_{>1}$ be the next existing instance after $\cI_1$ \Comment{$\cI_1$ was not deleted}
\If{$\cI_{>1}$ does not exist} continue to line~\ref{sliding:renumber} \EndIf
\If {$|S(\cI_{>1})| \geq L$}  \Comment{$|S(\cI_{>1})|$ is the number of items fed into $\cI_{>1}$}
\State Delete $\cI_1$
\EndIf
\State Renumber the instances and let $k$ be the number of remaining ones \label{sliding:renumber}
\If{$|S(\cI_1)| = L$} \textbf{report} $O_2(\cI_1)$ 
\Else{} \textbf{report} $O_2(\cI_2)$
\EndIf
\EndWhile
\end{algorithmic}
\end{algorithm*}

Let $e$ be the current item of the stream being processed by \Cref{alg:sliding} and let $E$ be the current sliding window consisting of the $L$ most recently processed items (including $e$).
While processing $e$, the algorithm first creates a new instance $\cI_{k+1}$ (called a bucket in Biabani et al.~\cite{Biabani21_WeightedMatching}) of $\cA$.
Then, $e$ is fed into all existing instances $\cI_1, \dots, \cI_{k+1}$.
Next, starting from the oldest instance $\cI_1$, only its newest similar instance, determined by $O_1$ (\Cref{lookahead:prop2} of \Cref{def:refined-lookahead}), is kept and every other instance in between is deleted. 
Whether a newest similar instance exists or not, the process is then repeated with the next oldest remaining instance until reaching the newest instance.
Note that the oldest and newest instances, $\cI_1$ and $\cI_{k+1}$ respectively, are never deleted by this process. 
However, if the number of items fed into the second oldest remaining instance $\cI_{>1}$ is at least $L$, i.e., the current sliding window $E$ is fully contained in the stream $S(\cI_{>1})$ of edges processed by $\cI_{>1}$, then $\cI_1$ is deleted.
The instances are then renumbered to $\cI_1,\dots,\cI_k$, from the oldest one to the newest, such that $k$ is the number of remaining instances. 
At this stage, the sliding window $E$ is sandwiched between streams $S(\cI_1)$ and $S(\cI_2)$.  
Finally, after processing the item, if the current sliding window contains exactly the edges processed by $\cI_1$, then the algorithm reports the second output $O_2$ of the instance $\cI_1$ as the solution, otherwise it reports $O_2(\cI_2)$.

In essence, the instances of $\cA$ created by \Cref{alg:sliding} simulate runs of a traditional streaming algorithm on suffixes of the current sliding window. 
Note that the oldest run always contains all items of the sliding window and potentially some additional ones.
The idea is to maintain runs on suffixes such that the value of $O_1$ of any two consecutive runs are not too different, while the value of $O_1$ of any non-consecutive runs are sufficiently different so as to ensure that at most a logarithmic number of instances of $\cA$ is used at any point of time. 

This idea is exactly captured when $\cA$, with two outputs $O_1$ and $O_2$, is a $(f, \alpha_1, \alpha_2, \beta)$-refined lookahead algorithm (which applies the smooth histogram technique by Braverman and Ostrovsky~\cite{Braverman2007SmoothHF}). 
The first output $O_1$ is used to determine how often a run on a suffix should be maintained, which depends on the smoothness criteria given by \Cref{lookahead:prop2} of \Cref{def:refined-lookahead}. 
The second output $O_2$ is a solution which, given the smoothness assumptions of the runs, is always guaranteed to be an $\alpha_2$-approximation of the next oldest run. 
We highlight that the smoothness assumptions are only guaranteed to hold for consecutive runs whose suffixes differ by more than one item.
Then, for a stream $S$ of items from a set $X$ and a non-decreasing function $f: 2^X \to \R^+$, the number of runs is at most logarithmic in $n$ as long as $\sigma_f(S) = f(S) / f_{\min}$, where $ f_{\min} = \min \{f(e) : e \in S\}$, is polynomial in $n$. 
We prove this formally in \Cref{thm:sliding}.

\thmsliding*

\begin{proof}
  We prove that \Cref{alg:sliding} satisfies the assertion of the theorem. 
  Let $\cA$ be the used $(f, \alpha_1, \alpha_2, \beta)$-refined lookahead algorithm with the outputs $O_1$ and $O_2$.
  
  \subparagraph{Approximation.} 
  Let $E$ be the sliding window at any instance of the algorithm, i.e., the set of the $L$ most recently processed items.
  The algorithm ensures that $E$ is sandwiched between streams of items fed to $\cI_1$ and $\cI_2$, i.e., $S_2 \subseteq E \subseteq S_1$ for $S_i = S(\cI_i), i \in \{1,2\}$. 
  We are now in one of two cases, either the items of $S_1$ and $S_2$ differ by exactly one item or more than one item. 
  
  In the former case, the algorithm asserts that $|S_2| < L$, otherwise $S_1$ would have been deleted, and therefore the items of $S_1$ are exactly those of the sliding window $E$, i.e., $|S_1| = L$. 
  The reported solution is then always $O_2(S_1) = O_2(E)$ which by \Cref{lookahead:prop2} of \Cref{def:refined-lookahead} (consider the case when $E = ABC = BC$) is trivially an $\alpha_2$-approximation of $f(E)$.
  
  In the latter case, the algorithm would have, at some point, deleted instances which caused $\cI_1$ and $\cI_2$ to become consecutive instances (Line~\ref{ln:DeleteSimilarInstances} of \Cref{alg:sliding}). 
  Consider the time $t^*$ when they first became adjacent. 
  Let $S_1^*$ and $S_2^*$ be the streams processed by $\cI_1$ and $\cI_2$, respectively, in the time $t^*$. 
  The algorithm asserts that $O_1(S_2^*) \geq (1 - \beta) \cdot O_1(S_1^*)$.
  Let $C$ be the remaining items fed into the instances such that $S_1 = S_1^*C$ and $S_2 = S_2^*C$.
  Then, by \Cref{lookahead:prop2} of \Cref{def:refined-lookahead} and $f$ being non-decreasing,
  \begin{align*}
  O_2(S_2) \leq f(S_2) \leq f(E) \leq f(S_1) \leq \alpha_2 \cdot O_2(S_2). \label{eq:B2appox}
  \end{align*}
  Hence, we have that, $O_2(S_2)$, is an $\alpha_2$-approximation of $f(E)$.
  Now, if $|S_1| \neq L$ the solution reported is $O_2(S_2)$, otherwise $|S_1| = L$ and the solution reported is $O_2(S_1) = O_2(E)$. 
  We conclude that in either case an $\alpha_2$-approximation of $f(E)$ is reported.
  
  \subparagraph{Space.}
  Let $k$ be the maximum number of instances stored by the algorithm after processing an item.
  After the process of deleting and renumbering the instances, the algorithm ensures that $O_1(\cI_{i+2}) < (1-\beta) \cdot O_1(\cI_i)$ holds for any instances $\cI_i$ and $\cI_{i+2}$.
  Thus for the largest odd number $k'$ not exceeding $k$, 
  \[ 
  (1 + \beta)^{\frac{k'-1}{2}} O_1(\cI_{k'}) < O_1(\cI_1).
  \]
  
  Recall that $\frac{f(S(\cI_1))}{f(S(\cI_{k'}))} \leq \sigma$. Then, by \Cref{lookahead:prop1} of \Cref{def:refined-lookahead}, we have that $\frac{O_1(\cI_1)}{O_1(\cI_{k'})} \leq \alpha_1 \sigma$. 
  It follows that 
  \[
  \frac{k'-1}{2} < \log_{1 + \beta} (\alpha_1 \sigma) 
  \text{~~~~~and~~~~~}
  k' = O\left(\frac{1}{\beta} \cdot \log (\alpha_1  \sigma)\right).
  \] 
  This implies the result since there are only ever $k+1 \leq k' + 2$ instances of $\cA$, each of which uses at most $s$ words of memory.
  \end{proof}

 A motivating example of the refined lookahead definition is exactly the Paz-Schwartzman algorithm {\mwm} with the first output $O_1$ as the sum of reduced weights $W'_S$, the second output $O_2$ as the weight of the returned matching $\w{\hat{M}(S)}$, and function $f$ as the weight of a maximum-weight matching $\textsf{MWM}(S)$.
 Now, consider the graph given in \Cref{sec:lb} (see \cref{fig:lookahead_example}).
We have that $\w{\hat{M}(AB)}= \w{\hat{M}(B)} = 2 + 2\varepsilon$, $W'_{AB} = 2 + 2\varepsilon$ and $W'_B = 1 + 2\varepsilon$.
We showed in \Cref{sec:lb} that this is indeed a hard instance for (standard) lookahead algorithms when the weight of the matching computed is used as the smoothness constraint (recall that $(1-\beta) \cdot \w{\hat{M}(AB)} \le \w{\hat{M}(B)}$ is then required in a hard instance, which is the case here). On the other hand, refined lookahead algorithms allow us to use the sum of reduced weights as the smoothness constraint. Since $(1-\beta) \cdot W'_{AB} \nleq W'_B$, for small enough $\beta$, the instance therefore is not hard for refined lookahead algorithms.


\end{document}